%% file: main.tex
\documentclass[11pt]{article}
\usepackage[T1]{fontenc}
\usepackage{lmodern}
\usepackage{amsmath,amssymb,amsthm}
\usepackage[margin=1in]{geometry}
\usepackage{booktabs,tabularx}
\usepackage{float,algpseudocode,authblk}
\floatstyle{ruled}
\newfloat{algorithm}{tbp}{loa}
\floatname{algorithm}{Algorithm}
\usepackage{microtype}
\usepackage{needspace}
\input{figures_tikz/style}
\input{figures_tikz/captions}

\usepackage[hidelinks]{hyperref}
\hypersetup{pdftitle={Online Covering with Maximum Delay under Subadditive Service Costs},
            pdfauthor={Anonymous authors}}
\newtheorem{theorem}{Theorem}[section]
\newtheorem{lemma}[theorem]{Lemma}
\newtheorem{proposition}[theorem]{Proposition}
\newtheorem{corollary}[theorem]{Corollary}
\theoremstyle{definition}
\newtheorem{definition}[theorem]{Definition}
\theoremstyle{remark}

\newcommand{\ALG}{\operatorname{ALG}}
\newcommand{\OPT}{\operatorname{OPT}}
\newcommand{\E}{\mathbb E}
\newcommand{\Prob}{\mathbb P}
\newcommand{\ind}{\mathbf 1}
\newcommand{\Rrho}{R_{\rho}}
\newcommand{\supp}{\operatorname{supp}}
\newcommand{\Span}{\operatorname{span}}
\newcommand{\Threshold}{\textnormal{\textsc{Threshold}}}
\newcommand{\Height}{\textnormal{\textsc{Height}}}
\title{Online Covering with Maximum Delay\\under Subadditive Service Costs}
\author{
Tianhang Lu\thanks{The first two authors contribute equally.},
Runtian Ren,
Shengcai Liu
}
\affil{
Guangdong Provincial Key Laboratory of Brain-Inspired Intelligent Computation,\\
Department of Computer Science and Engineering,\\
Southern University of Science and Technology, Shenzhen 518055, China\\
liusc3@sustech.edu.cn
}
\date{}

\begin{document}
\maketitle

\begin{abstract}
We study online covering in which each instantaneous service pays its
purchase cost and one maximum waiting time, with no effect on future
requests. For static realizable services, monotone subadditivity suffices
for optimal competitive ratios; submodularity is unnecessary. A normalized
monotone subadditive lower-bound oracle with realization factor $\rho$
yields ratios $\rho+1$ deterministically and $1/(1-e^{-1/\rho})$ randomly
against an oblivious adversary. Exact batch optimization gives the optimal
constants $2$ and $e/(e-1)$. The randomized algorithm uses one global
threshold on a seed-independent virtual-height trajectory, whose active
time is a lower bound on the offline optimum. Weighted vertex cover gives
a strict separation from submodularity on a three-edge bipartite path,
with polynomial-time batch implementations through min-cut and LP
rounding. An offline consecutive-batch normal form also transfers static
approximation guarantees to the offline problem.
\end{abstract}

\input{sections_subadditive/intro}
\input{sections_subadditive/model_offline}
\input{sections_subadditive/algorithms}
\input{sections_subadditive/lower_bounds}
\input{sections_subadditive/applications}

\section{Concluding Remarks}
\label{sec:conclusion}

For static services with one maximum-delay charge per joint purchase, monotone subadditivity suffices for the exact competitive constants $2$ and $e/(e-1)$. Vertex cover is both an application and a witness that the theorem goes beyond submodular batch costs. 
The offline normal form and the realizable-oracle formulation separate temporal uncertainty from the computational task of optimizing a batch.
This separation leaves a concrete computational question: can an efficient algorithm for general weighted graphs improve on $3$ deterministically or $\sqrt e/(\sqrt e-1)$ randomly? The tight examples in the appendix concern our LP-rounding implementations and leave this gap open for other efficient algorithms.

\appendix
\input{sections_subadditive/tightness}

\bibliographystyle{plain}
\bibliography{references}
\end{document}

%% file: figures_tikz/style.tex
\usepackage{tikz}
\usetikzlibrary{arrows.meta,calc,decorations.pathreplacing,patterns,positioning}
\definecolor{vcblue}{HTML}{0072B2}
\definecolor{vcorange}{HTML}{D55E00}
\definecolor{vcink}{HTML}{263B50}
\definecolor{vcmuted}{HTML}{81909B}
\definecolor{vclight}{HTML}{EDF3F6}
\tikzset{
  vc figure/.style={font=\small, line cap=round, line join=round, text=vcink},
  vc axis/.style={-{Stealth[length=4pt]}, draw=vcink, line width=.55pt},
  vc guide/.style={draw=vcmuted!45, densely dashed, line width=.4pt},
  vc vertex/.style={circle, draw=vcink, fill=white, line width=.65pt,
                    minimum size=5.5pt, inner sep=0pt},
  vc chosen/.style={vc vertex, fill=vcblue, draw=vcblue},
  vc active/.style={draw=vcink, line width=1.2pt},
  vc inactive/.style={draw=vcmuted!45, dashed, line width=.65pt},
  vc service/.style={rectangle, draw=vcink, fill=vcink,
                     minimum size=4.7pt, inner sep=0pt},
  vc arrival/.style={circle, fill=vcink, draw=vcink,
                     minimum size=3.7pt, inner sep=0pt},
  vc panel/.style={font=\small\bfseries, anchor=west},
  vc note/.style={font=\footnotesize, text=vcink},
}

%% file: figures_tikz/captions.tex
\newcommand{\CaptionSeparation}{%
  A subadditive batch cost need not be submodular, even on a bipartite graph.
  Solid edges are the requested types; dashed edges are absent from the
  queried batch. Filled vertices give minimum covers. In the union, the
  two disjoint end edges force two selected vertices. The four values
  violate the submodular inequality. All vertex costs are one.}
\newcommand{\CaptionOffline}{%
  The consecutive-batch normal form. A horizontal segment is one original
  service's waiting interval, not a request-by-request delay sum; its
  right square is the service time. Overlapping intervals belong to one
  connected component. Each component becomes a full batch, served at
  its last arrival (times $2$ and $5$ here). Subadditivity bounds the new
  purchase cost, and the batch span is bounded by the sum of the old
  interval lengths. Here $p_\ell$ and $d_\ell$ are the original purchase
  and delay costs.}
\newcommand{\CaptionHeights}{%
  One virtual trajectory, two actual schedules. Both requests have the
  same type with service cost one and arrive at $0$ and $1/2$. The first
  block pauses while the new lowest block rises; they merge at time $1$
  and height $1/2$. Hollow circles mark threshold crossings. The two
  displayed thresholds are alternative realizations of the single
  global random choice, not fresh samples at successive services.
  An actual service does not reset the heights. The virtual active time
  is $L=3/2$, independently of the threshold.}
\newcommand{\CaptionPotential}{%
  Why virtual active time is measured by a potential decrease. Older
  blocks have higher heights. Between the lowest height $h$ and the next
  block height, the layer $A_x$ contains only $B_3$. Raising $B_3$ by
  $dh$ removes the hatched strip of width $g(B_3)$ and height $dh$;
  every other layer is unchanged. At speed $1/g(B_3)$ the area decreases
  at unit rate. The plotted widths are schematic monotone layer costs,
  not an assumption of additivity or strict cost increase.}

%% file: sections_subadditive/intro.tex
\section{Introduction}
\label{sec:intro}

Online covering couples a timing decision with a static optimization
problem. Waiting allows requests to share a purchase, but incurs delay.
We study the objective in which a service at time $t$, purchasing $S$ and
serving a nonempty batch $B$ with arrival times $a_i$, costs
\[
    c(S)+\max_{i\in B}(t-a_i).
\]
The objective sums these costs over services. Purchases are instantaneous,
serve only pending requests, and have no effect on future services.

Which properties of the joint purchase cost are enough for optimal online
batching? We show that normalized, monotone subadditive costs support the
optimal competitive constants $2$ and $e/(e-1)$ under exact batch
optimization. Diminishing returns are not required. When batch
optimization is approximate, the competitive guarantees depend only on
the realization factor of a suitable static batch oracle.

Weighted vertex cover is a natural test case. Requests are edges of a
known graph, and a service buys a vertex set and clears all pending
requests incident to it. Its minimum batch cost is subadditive because
the union of two covers covers the union of their edge sets. Yet it need
not be submodular, even on a three-edge bipartite path. Vertex cover thus
provides both an application and a strict separation of the cost
assumptions.

The companion manuscript on multi-level aggregation with maximum
delay~\cite{lu2026mla} obtains the same constants for static realizable
submodular joint costs through DP-Envelope. Here, direct interval charging
and a virtual-height process establish the bounds under subadditivity
and accommodate approximate batch optimization. Section~\ref{sec:mla}
gives the formal relationship between the two frameworks.

\subsection{Results and scope}

Let $\kappa(A)$ be the minimum purchase cost of one service covering request
types $A$. The service model is static: such a service can be executed
whenever precisely these types are pending, and executing it does not
change the feasibility or cost of subsequent services. We assume that
$\kappa$ is normalized, monotone, and subadditive. An oracle supplies a
lower bound $g(A)\le\kappa(A)$ and an executable service of cost at most
$\rho g(A)$, where $\rho\ge1$ is fixed. The exact oracle has
$g=\kappa$ and $\rho=1$.

\begin{theorem}[Oracle framework; informal]
\label{thm:intro}
A normalized monotone lower-bound oracle with realization factor $\rho$
gives a strictly $(\rho+1)$-competitive deterministic algorithm.
If $g$ is also subadditive, a randomized algorithm is strictly
\[
        \Rrho:=\frac{1}{1-e^{-1/\rho}}
\]
competitive against an oblivious adversary. Both algorithms are causal,
need no notification that the input has ended, and use linearly many
batch-oracle calls in the number of arrival epochs.
\end{theorem}

The oracle separates the two tasks: $g$ certifies the timing decisions,
while its realization supplies the purchase. The required monotone
lower-bound function is part of the oracle contract, not a property
guaranteed by an arbitrary approximation algorithm.

For exact batch costs, both ratios are optimal on every fixed system
containing a repeatable positive-cost type, even with a fixed additive
constant. The lower bounds use just that type, and hence apply to a
single-edge graph. Against feedback-dependent arrivals, the optimal
randomized ratio is instead $2$ when compared with the expected
hindsight optimum of the realized input.

Table~\ref{tab:results} gives the principal implementations. On general
graphs, an exact weighted vertex-cover oracle is computationally expensive;
the fractional relaxation gives the stated polynomial-time bounds. The
general-graph, hypergraph, and greedy set-cover bounds are tight or
asymptotically tight for the specified implementations, rather than lower
bounds for every polynomial-time algorithm.

\begin{table}[t]
\centering\small
\begin{tabularx}{\textwidth}{@{}>{\raggedright\arraybackslash}Xccc@{}}
\toprule
Service system / batch oracle & $\rho$ & Deterministic & Randomized\\
\midrule
Static subadditive / exact & $1$ & $2$ & $e/(e-1)$\\
Weighted bipartite VC / min-cut & $1$ & $2$ & $e/(e-1)$\\
Weighted general VC / LP rounding & $2$ & $3$ & $\sqrt e/(\sqrt e-1)$\\
Rank-$r$ hypergraph VC / LP rounding & $r$ & $r+1$ & $\Rrho\ (\rho=r)$\\
Family-service Set Cover / LP + greedy & $H_m$ & $H_m+1$ & $\Rrho\ (\rho=H_m)$\\
\bottomrule
\end{tabularx}
\caption{Strict competitive guarantees against an oblivious adversary
in the randomized case. Here $R_\rho=(1-e^{-1/\rho})^{-1}$,
$m$ is the number of request types, and $H_m=\sum_{j=1}^m1/j$.
All rows except the unrestricted exact-oracle row
have polynomial-time batch implementations for explicit rational input.
Continuous-time randomized execution is in the real-arithmetic model
described in Section~\ref{sec:execution}.}
\label{tab:results}
\end{table}

Offline, subadditivity permits a normal form consisting of full services of
consecutive arrival epochs, each performed at its final arrival time. This
gives an exact interval dynamic program using quadratically many batch
evaluations. Replacing each batch optimum by a feasible
$\rho$-approximation gives a $\rho$-approximation offline, without needing
monotonicity of the approximation routine's output values. In particular,
offline weighted vertex cover with maximum delay is polynomial-time
solvable on bipartite graphs and NP-hard on general graphs.

\subsection{Algorithms and analysis}

The deterministic algorithm serves the entire pending queue when its
oldest age reaches $g(A)$. Monotonicity makes the deadline causal.
An interval-charging argument assigns each online waiting interval to
offline purchases and offline waiting time, proving that the sum of
online waiting times is at most the offline optimum.

The randomized algorithm has a different construction. It maintains a
deterministic virtual process in which consecutive blocks have heights in
$[0,1]$, ordered from high to low by arrival time. Only the lowest block
rises, at speed $1/g(B)$, merging with an older block when their heights
meet. A single random threshold determines actual full services. An
integral of layer costs decreases at unit rate whenever the virtual
process is active. Subadditivity bounds its insertion increments over any
consecutive batch by that batch's oracle cost plus its time span. Thus the
total virtual active time is a lower bound on $\OPT$. An exponential
threshold distribution bounds the expected purchase and delay cost
together by $\Rrho$ per unit of active time.

This proof uses no diminishing-returns inequality. A concrete witness to
the extra generality is the unit-cost path with consecutive edges
$e_1,e_2,e_3$. For $A=\{e_1,e_2\}$ and $B=\{e_2,e_3\}$,
\[
 \kappa(A)=\kappa(B)=\kappa(A\cap B)=1,
 \qquad \kappa(A\cup B)=2.
\]
Hence $\kappa(A)+\kappa(B)<\kappa(A\cup B)+\kappa(A\cap B)$.
The failure of submodularity occurs already on a bipartite graph, where the exact oracle is efficiently computable.
Figure~\ref{fig:separation} below also displays the four queried edge sets and minimum-cover witnesses for this separation.

\begin{figure}[htbp]
\centering
\input{figures_tikz/non_submodular}
\caption{\CaptionSeparation}
\label{fig:separation}
\end{figure}

\subsection{Related work}

Dynamic TCP acknowledgment is a basic model of the purchase--delay tradeoff.
Dooly et al.~\cite{dooly2001} analyze both request-additive
delay and the sum of per-batch maximum delays, obtaining the optimal
deterministic factor $2$ in the latter one-type setting. Karlin et al.~\cite{karlin2003} develop optimal randomized acknowledgment
algorithms for the classical request-additive objective.
Albers and Bals~\cite{albers2005} study acknowledgment cost plus a single
global maximum delay. Here we instead sum one maximum over each service.
Bhore et al.~\cite{bhore2026} study more general
batch-aware and batch-oblivious acknowledgment delays. Their general
formulation allows a broader range of delay functions; here we fix the
per-service maximum and study the joint covering cost.

In Multi-level Aggregation (MLA) with delay, a rooted connected subtree serves requests located at its vertices. 
The classical request-additive model was studied, among others, by Bienkowski et al.~\cite{bienkowski2020}. 
Broader frameworks for optimization with delay and deadlines include the work of Azar and Touitou~\cite{azar2019}. 
The closest maximum-delay comparison is the companion submodular framework discussed above and in Section~\ref{sec:mla}.

Set Cover with Delay, studied by Azar et al.~\cite{azar2020}, charges delay request by request, and also includes a $3$-competitive result for vertex cover in that model. 
Our objective instead charges one maximum per joint service. 
The set-cover corollary allows an entire family of sets to be bought as one service; it does not
cover atomic purchases with separate maximum-delay charges (Section~\ref{sec:boundary}).

The spatial implementations use standard covering relaxations,
min-cut duality, and greedy set cover~\cite{williamson2011,schrijver2003,chvatal1979}.
Their batch values and feasible realizations supply the oracle contract.

\paragraph{Organization.}
Section~\ref{sec:model} specifies the service and oracle contracts.
Section~\ref{sec:offline} proves the offline normal form and dynamic program.
Sections~\ref{sec:det} and~\ref{sec:random} give the online algorithms.
Section~\ref{sec:lower} proves the matching lower bounds.
Sections~\ref{sec:applications} and~\ref{sec:mla} discuss implementations and
the relation to MLA. Section~\ref{sec:boundary} describes the service-model
boundary. Appendix~\ref{sec:tight} gives implementation-specific tightness
constructions.

%% file: figures_tikz/non_submodular.tex
\begin{tikzpicture}[vc figure, x=1cm, y=1cm]
  \node[vc note, anchor=west] at (0,1.18)
    {The same unit-cost path in every panel};
  \node[vc chosen] at (9.9,1.18) {};
  \node[vc note, anchor=west] at (10.05,1.18) {selected cover vertex};

  \begin{scope}[xshift=0cm]
    \node[vc panel] at (0,.62) {(a) $A=\{e_1,e_2\}$};
    \draw[vc active] (0,0)--(1,0)--(2,0);
    \draw[vc inactive] (2,0)--(3,0);
    \foreach \i in {0,2,3} \node[vc vertex] at (\i,0) {};
    \node[vc chosen] at (1,0) {};
    \foreach \i in {0,1,2,3} \node[vc note, below=4pt] at (\i,0) {$v_{\i}$};
    \foreach \i/\j in {.5/1,1.5/2,2.5/3}
      \node[vc note, above=3pt] at (\i,0) {$e_{\j}$};
    \node at (1.5,-.82) {$\kappa(A)=1$};
  \end{scope}
  \begin{scope}[xshift=3.8cm]
    \node[vc panel] at (0,.62) {(b) $B=\{e_2,e_3\}$};
    \draw[vc inactive] (0,0)--(1,0);
    \draw[vc active] (1,0)--(2,0)--(3,0);
    \foreach \i in {0,1,3} \node[vc vertex] at (\i,0) {};
    \node[vc chosen] at (2,0) {};
    \foreach \i in {0,1,2,3} \node[vc note, below=4pt] at (\i,0) {$v_{\i}$};
    \foreach \i/\j in {.5/1,1.5/2,2.5/3}
      \node[vc note, above=3pt] at (\i,0) {$e_{\j}$};
    \node at (1.5,-.82) {$\kappa(B)=1$};
  \end{scope}
  \begin{scope}[xshift=7.6cm]
    \node[vc panel] at (0,.62) {(c) $A\cup B$};
    \draw[vc active] (0,0)--(1,0)--(2,0)--(3,0);
    \foreach \i in {0,3} \node[vc vertex] at (\i,0) {};
    \foreach \i in {1,2} \node[vc chosen] at (\i,0) {};
    \foreach \i in {0,1,2,3} \node[vc note, below=4pt] at (\i,0) {$v_{\i}$};
    \foreach \i/\j in {.5/1,1.5/2,2.5/3}
      \node[vc note, above=3pt] at (\i,0) {$e_{\j}$};
    \node at (1.5,-.82) {$\kappa(A\cup B)=2$};
  \end{scope}
  \begin{scope}[xshift=11.4cm]
    \node[vc panel] at (0,.62) {(d) $A\cap B$};
    \draw[vc inactive] (0,0)--(1,0) (2,0)--(3,0);
    \draw[vc active] (1,0)--(2,0);
    \foreach \i in {0,2,3} \node[vc vertex] at (\i,0) {};
    \node[vc chosen] at (1,0) {};
    \foreach \i in {0,1,2,3} \node[vc note, below=4pt] at (\i,0) {$v_{\i}$};
    \foreach \i/\j in {.5/1,1.5/2,2.5/3}
      \node[vc note, above=3pt] at (\i,0) {$e_{\j}$};
    \node at (1.5,-.82) {$\kappa(A\cap B)=1$};
  \end{scope}

  \node[fill=vclight, rounded corners=2pt, inner xsep=12pt, inner ysep=6pt]
    at (7.2,-2.15)
    {$\underbrace{\kappa(A)+\kappa(B)}_{1+1=2}
      \;<\;\underbrace{\kappa(A\cup B)+\kappa(A\cap B)}_{2+1=3}$};
\end{tikzpicture}

%% file: sections_subadditive/model_offline.tex
\section{Static Services and Batch Oracles}
\label{sec:model}

\subsection{Service model}

Let $U$ denote a finite known universe of request types. 
A finite static family of actions is given; each action $S$ has a fixed nonnegative cost $c(S)$ and a fixed subset of types that it covers. 
Every type is coverable. 
An input consists of finitely many occurrences $(a_i,u_i)$, with $a_i\ge0$ and $u_i\in U$. 
Types may repeat. 
An action at time $t$ instantaneously serves every pending occurrence of each type it covers. 
If the nonempty set of occurrences served is $B$, it costs
\begin{equation}
    c(S)+t-\min_{i\in B}a_i.
\label{eq:servicecost}
\end{equation}
Actions serving no requests can be omitted. 
Every request must eventually be served in finite time. 
All arrivals at a timestamp are revealed before services at that timestamp; multiple instantaneous services are allowed.
There is no persistent acquisition, action duration, capacity constraint, time-dependent eligibility, or effect on the cost of subsequent actions.

For $A\subseteq U$, let $\kappa(A)$ denote the minimum cost of a single action covering all of $A$, and set $\kappa(\varnothing)=0$. 
We require that this minimum is finite and attained for every nonempty $A$. 
In particular, an action attaining $\kappa(A)$ can be executed whenever the pending types are precisely $A$. 
Every action serving occurrences of types $A$ has purchase cost at least $\kappa(A)$. 
Monotonicity follows from this definition; the substantive spatial assumption is
\begin{equation}
 \kappa(A\cup B)\le\kappa(A)+\kappa(B)
       \qquad(A,B\subseteq U).
 \label{eq:kappasub}
\end{equation}
A sufficient way to ensure~\eqref{eq:kappasub} is to allow two actions to be combined into one action that covers their union, at purchase cost no greater than the sum. 
The proofs need the scalar inequality and attainability, not a particular union operation on representations.

For occurrence sets we write $\kappa(B)$ as shorthand for
$\kappa(\supp(B))$, where $\supp(B)$ is the set of their types. 
We use the same convention for other batch functions. 
Thus repeated requests matter for arrivals and waiting, but not for a batch's purchase cost. 
An offline schedule is allowed to serve only part of the pending queue. 
Our online algorithms use \emph{full services}, which clear the entire queue.

Zero-cost types can be removed before applying the algorithms. 
Let $Z=\{u:\kappa(\{u\})=0\}$. 
Monotonicity and subadditivity give $\kappa(A\cup Z)=\kappa(A)$ for every $A$. 
Serve each arriving cohort of zero-cost types immediately using a zero-cost joint action. 
This adds no cost. 
Conversely, deleting these requests cannot increase the optimum.
The two instances have the same optimum, and the oracle on the remaining types is simply its restriction. 
Henceforth $\kappa(A)>0$ for every nonempty batch, unless the input is empty, which has cost zero.

Let $b_1 < \cdots < b_n$ denote the distinct remaining arrival times and $C_j$ the nonempty cohort arriving at $b_j$. 
We use $n$ for the number of arrival epochs and $N$ for the number of occurrences. Define
\[
 B_{j:k} = C_j\cup\cdots\cup C_k,
 \qquad \Span(B_{j:k}) = b_k-b_j.
\]
Unions of cohorts here are unions of occurrences, not identifications of repeated types. 
A partition into consecutive batches always cuts between complete epochs, never within simultaneous arrivals.

\subsection{Oracle and competitive guarantees}

\begin{definition}[Realizable lower-bound oracle]
\label{def:oracle}
For a fixed $\rho\ge1$, an oracle consists of a deterministic function $g:2^U\to\mathbb R_{\ge0}$ and a deterministic realization $S_g(A)$ such that $g(\varnothing)=0$ and, for every nonempty $A$,
\begin{equation}
 0<g(A)\le\kappa(A)\le c(S_g(A))\le\rho g(A).
 \label{eq:oracle}
\end{equation}
The action $S_g(A)$ covers $A$. We require $g$ to be monotone, and for randomized algorithms additionally require $g(A\cup B)\le g(A)+g(B)$. 
Ties within a batch routine are resolved by a fixed rule.
\end{definition}

The lower bound and the realization have different roles: $g$ controls time and certifies comparison with $\OPT$, whereas $S_g$ is the action actually purchased. 
An arbitrary approximate-cover output need not have a monotone or subadditive value. 
Fractional optimum values, followed by rounding, provide useful examples satisfying the full contract.

For a fixed finite input $\sigma$, let $\OPT(\sigma)$ denote its minimum offline cost. 
A strict deterministic ratio $R$ means
$\ALG(\sigma)\le R\OPT(\sigma)$ for every $\sigma$. 
For an oblivious adversary the input is fixed independently of the random seed, and the
randomized guarantee is $\E[\ALG(\sigma)]\le R\OPT(\sigma)$. 
Lower bounds below also exclude guarantees with an arbitrary fixed additive constant. Only Section~\ref{sec:adaptive} considers feedback-dependent inputs.

\section{Offline Structure and Computation}
\label{sec:offline}

\begin{lemma}[Consecutive-batch normal form]
\label{lem:normal}
Every feasible schedule can be replaced, without increasing its cost, by a sequence of full services. 
Their batches form a partition of the arrival epochs into consecutive intervals, and a batch $B_{j:k}$ is served at time $b_k$ at cost $\kappa(B_{j:k})+b_k-b_j$.
\end{lemma}

\begin{proof}
Suppose there are at most $N$ non-empty services, each serving a previously unserved occurrence. 
A service at time $t_\ell$ serving $D_\ell$ defines the closed waiting interval
\[
 J_\ell=[\min_{i\in D_\ell}a_i,t_\ell].
\]
Its length is exactly that service's delay charge. 
Group the services by connected components of $\bigcup_\ell J_\ell$, including singleton components. 
Consider one component $[b,d]$ and let $B$ denote all occurrences served by its services. 
Its earliest arrival is $b$. 
Every occurrence arriving in $[b,d]$ belongs to this component: its own service interval contains its arrival and hence intersects $[b,d]$. 
Thus the components partition whole arrival epochs into consecutive groups. 
In particular, simultaneous occurrences cannot lie in different groups.
Let $r$ denote the last arrival in $B$. 
Since the intervals cover $[b,d]$,
\[
 r-b\le d-b\le\sum_{\ell\text{ in component}}|J_\ell|.
\]
Repeated subadditivity and the lower bound on every actual purchase give
\[
 \kappa(B)\le\sum_{\ell\text{ in component}}\kappa(D_\ell) \le\sum_{\ell\text{ in component}}c(S_\ell).
\]
Replace the component's services by one minimum-cost full service at $r$.
Its cost is no greater than the original component cost. Components are temporally disjoint, so performing these replacements in order is feasible: exactly the current component's occurrences are pending at its replacement time. 
This proves the assertion, including zero-length components.
\end{proof}

\begin{figure}[htbp]
\centering
\input{figures_tikz/offline_normal_form}
\caption{\CaptionOffline}
\label{fig:offline}
\end{figure}
Figure~\ref{fig:offline} illustrates both operations in the proof: combining services within a waiting-interval component and advancing the replacement service to its final arrival.

Conversely, every partition into consecutive complete epochs can be executed by these minimum-cost actions, because the next epoch has not arrived when a batch is served. Consequently
\begin{equation}
 \OPT=\min_{\mathcal P}\sum_{B\in\mathcal P} \bigl(\kappa(B)+\Span(B)\bigr),
 \label{eq:normal}
\end{equation}
where $\mathcal P$ ranges over consecutive-epoch partitions. 
In particular, a minimum exists; no compactness argument about arbitrary service times is needed.

\begin{theorem}[Offline dynamic program]
\label{thm:offline}
Set $F(0)=0$ and, for $1\le k\le n$, set
\begin{equation}
 F(k)=\min_{1\le j\le k}
       \left\{F(j-1)+\kappa(B_{j:k})+b_k-b_j\right\}.
 \label{eq:dp}
\end{equation}
Then $F(n)=\OPT$. The recurrence requires $O(n^2)$ batch evaluations and arithmetic comparisons, apart from representing the queried batches.
An optimal schedule is recovered by predecessor pointers and batch realizations.
\end{theorem}

\begin{proof}
In~\eqref{eq:normal}, the final batch of a prefix ending at epoch $k$ is $B_{j:k}$ for a unique $j$. 
Its preceding batches optimally partition the first $j-1$ epochs. 
Minimizing over $j$ gives~\eqref{eq:dp}; induction on $k$ proves the value and reconstruction claims.
\end{proof}

\begin{proposition}[Offline approximation transfer]
\label{prop:offlineapprox}
Suppose a batch routine returns a feasible action of value $\widehat c(B)$ with $\kappa(B)\le\widehat c(B)\le\rho\kappa(B)$. 
Substituting $\widehat c$ for $\kappa$ in~\eqref{eq:dp} yields a feasible offline $\rho$-approximation. 
The output-value function $\widehat c$ need not be monotone or subadditive.
\end{proposition}

\begin{proof}
Take an optimal partition $\mathcal P^*$ from~\eqref{eq:normal}. Its modified value is at most
\[
 \sum_{B\in\mathcal P^*}\left(\rho\kappa(B)+\Span(B)\right)
 \le\rho\OPT.
\]
The modified dynamic program chooses a partition of no larger modified value. 
Execute its returned actions at the final arrival of each batch.
Each batch comprises complete epochs, so the resulting services are feasible and their cost equals the modified value. 
This last statement would not hold for an arbitrary partition cutting simultaneous arrivals; the epoch convention is part of the recurrence.
\end{proof}

\begin{proposition}[Complexity on graphs]
\label{prop:offlinecomplexity}
Offline weighted vertex cover with maximum delay is NP-hard on general graphs, even when all requests arrive at time zero and all vertex costs are one. 
It is polynomial-time solvable on bipartite graphs with rational costs and arrival times, and has a polynomial-time $2$-approximation on general graphs.
\end{proposition}

\begin{proof}
Given a graph, request each edge at time zero. A minimum vertex cover served at time zero has cost equal to the static optimum. 
Conversely, the union of vertices purchased by any feasible schedule covers every edge, and its weight is at most that schedule's total purchase cost, hence at most its total cost. 
The two optima are therefore equal, proving the reduction from static Vertex Cover~\cite{williamson2011}. 
Only the replacement service at time zero is asserted to have zero delay; an arbitrary original schedule need not have zero delay.

On bipartite graphs each $\kappa(B_{j:k})$ is a polynomial-time minimum-cut computation, as detailed in Section~\ref{sec:vc}. 
Apply Theorem~\ref{thm:offline}. 
On general graphs use the standard factor-two fractional-cover rounding and Proposition~\ref{prop:offlineapprox}.
\end{proof}

%% file: figures_tikz/offline_normal_form.tex
\begin{tikzpicture}[vc figure, x=1.7cm, y=1cm]
  \node[vc panel] at (-1.5,3.62) {(a) Original waiting intervals};
  \fill[vcblue!5] (0,.66) rectangle (3,3.18);
  \fill[vcorange!6] (4,.66) rectangle (6,3.18);
  \node[vc note] at (1.5,3.0) {component $[0, 3]$};
  \node[vc note] at (5,3.0) {component $[4, 6]$};
  \foreach \x in {0,1,2,3,4,5,6} \draw[vc guide] (\x,.7)--(\x,2.73);

  \node[vc note,anchor=east] at (-.14,2.35) {$D_1=\{r_2\}$};
  \draw[vcblue,line width=1.4pt] (1,2.35)--(2.5,2.35);
  \node[vc arrival,fill=vcblue,draw=vcblue] at (1,2.35) {};
  \node[vc service,fill=vcblue,draw=vcblue] at (2.5,2.35) {};
  \node[vc note,above=2pt] at (1.75,2.35) {$d_1=3/2$};

  \node[vc note,anchor=east] at (-.14,1.68) {$D_2=\{r_1,r_3\}$};
  \draw[vcblue,line width=1.4pt] (0,1.68)--(3,1.68);
  \node[vc arrival,fill=vcblue,draw=vcblue] at (0,1.68) {};
  \node[vc service,fill=vcblue,draw=vcblue] at (3,1.68) {};
  \node[vc note,above=2pt] at (1.5,1.68) {$d_2=3$};

  \node[vc note,anchor=east] at (-.14,1.01) {$D_3=\{r_4,r_5\}$};
  \draw[vcorange,line width=1.4pt] (4,1.01)--(6,1.01);
  \node[vc arrival,fill=vcorange,draw=vcorange] at (4,1.01) {};
  \node[vc service,fill=vcorange,draw=vcorange] at (6,1.01) {};
  \node[vc note,above=2pt] at (5,1.01) {$d_3=2$};

  \draw[vc axis] (-.12,.2)--(6.38,.2) node[right,vc note] {$t$};
  \node[vc note,anchor=east] at (-.14,.2) {arrivals};
  \foreach \x/\r in {0/1,1/2,2/3,4/4,5/5} {
    \node[vc arrival] at (\x,.2) {};
    \node[vc note,above=4pt] at (\x,.2) {$r_{\r}$};
  }
  \foreach \x in {0,1,2,3,4,5,6}
    \node[vc note,below=4pt] at (\x,.2) {$\x$};

  \node[vc panel] at (-1.5,-.85) {(b) Full batches served at their last arrivals};
  \begin{scope}[yshift=-.65cm]
  \fill[vcblue!12] (0,-1.58) rectangle (2,-1.25);
  \fill[vcorange!15] (4,-1.58) rectangle (5,-1.25);
  \draw[vc axis] (-.12,-1.58)--(6.38,-1.58) node[right,vc note] {$t$};
  \draw[vcblue,line width=1pt] (0,-1.25)--(2,-1.25);
  \draw[vcorange,line width=1pt] (4,-1.25)--(5,-1.25);
  \foreach \x/\r in {0/1,1/2,2/3,4/4,5/5}
    \node[vc arrival] at (\x,-1.58) {};
  \node[vc service,fill=vcblue,draw=vcblue] at (2,-1.25) {};
  \node[vc service,fill=vcorange,draw=vcorange] at (5,-1.25) {};
  \foreach \x in {0,1,2,3,4,5,6}
    \node[vc note,below=4pt] at (\x,-1.58) {$\x$};
  \node[vc note,above=3pt] at (1,-1.25) {$B_{1:3}=\{r_1,r_2,r_3\}$};
  \node[vc note,above=3pt] at (4.5,-1.25) {$B_{4:5}=\{r_4,r_5\}$};
  \node[vc note,align=center] at (1,-2.38)
    {$\kappa(B_{1:3})\le p_1+p_2$\\$\mathrm{span}(B_{1:3})=2\le d_1+d_2$};
  \node[vc note,align=center] at (4.5,-2.38)
    {$\kappa(B_{4:5})\le p_3$\\$\mathrm{span}(B_{4:5})=1\le d_3$};
  \end{scope}
  \node[vc service] at (5.2,3.64) {};
  \node[vc note,anchor=west] at (5.31,3.64) {service};
\end{tikzpicture}

%% file: sections_subadditive/algorithms.tex
\section{A Deterministic Threshold Algorithm}
\label{sec:det}

Let $A(t)$ denote the current pending occurrences. 
When it is nonempty, write $b = \min_{i\in A(t)}a_i$ and $W(t) = t-b$. 
The algorithm serves $A(t)$ using $S_g(A(t))$ as soon as
\begin{equation}
 W(t)=g(A(t)).
 \label{eq:dettrigger}
\end{equation}
An arrival is processed before testing the trigger at the same timestamp.
Between arrivals the queue does not change until a service. 
Since $g$ is monotone, an arrival cannot decrease the threshold. 
Thus a timer set to $b+g(A)$ is never moved into the past. 
If no further request arrives, the timer is finite. 
This specifies a causal algorithm without requiring an end-of-input signal.

\begin{algorithm}[htbp]
\caption{\Threshold$(g,S_g)$}
\label{alg:det}
\begin{algorithmic}[1]
\State Initialize the pending queue $A\gets\varnothing$.
\State On an arrival epoch, insert its entire cohort into $A$.
\State Whenever $A\ne\varnothing$, let $b$ denote its oldest arrival and set the timer to $b+g(A)$.
\State Process all arrivals tied with the timer and recompute it.
\State If the timer is due, execute $S_g(A)$ and set $A\gets\varnothing$.
\end{algorithmic}
\end{algorithm}

\begin{lemma}[Waiting-time certificate]
\label{lem:detcertificate}
Let $W_1,\ldots,W_s$ denote the maximum delays of the full services performed by \Threshold. Then $\sum_i W_i\le\OPT$.
\end{lemma}

\begin{proof}
The online waiting intervals $I_i=[b_i,t_i)$, from each queue's oldest arrival to its service time, are pairwise disjoint. 
Their lengths are $W_i$. 
At every $y\in I_i$, after processing arrivals at $y$, let $A_i(y)$ denote the requests accumulated in this interval. 
Since the service has not occurred,
\begin{equation}
 \kappa(A_i(y)) \ge g(A_i(y)) > y-b_i.
 \label{eq:pretrigger}
\end{equation}

Fix any feasible offline schedule. 
Write $p_\ell$ for its purchase cost, $s_\ell$ for its service time, and $d_\ell$ for its maximum-delay charge.
Its waiting intervals are $J_\ell=[s_\ell-d_\ell,s_\ell]$; set $V=\bigcup_\ell J_\ell$. 
For an online interval define
\[
 Q_i=\sum_{\ell:s_\ell\in I_i}p_\ell,
 \qquad T_i=|V\cap I_i|.
\]
We show $W_i\le Q_i+T_i$. 
If $I_i\subseteq V$, then $T_i=W_i$.
Otherwise put $x=\sup(I_i\setminus V)$ and take any
$y\in I_i\setminus V$. 
Every request in $A_i(y)$ was already served offline by time $y$: a later service of such a request would have a waiting interval containing $y$. 
Its offline service cannot precede $b_i$, since the request itself arrives at or after $b_i$. 
The offline services in $[b_i,y]$ therefore cover all of $A_i(y)$. 
Subadditivity and the purchase-cost lower bound imply
\[
 Q_i\ge\kappa(A_i(y))>y-b_i.
\]
Taking a sequence of such $y$ tending to $x$ gives $Q_i\ge x-b_i$.
The portion $(x,t_i)$ is contained in $V$, so $T_i\ge t_i-x$.
Together these yield the claim. 
Endpoint conventions do not affect lengths or the limit argument.

Disjointness of the $I_i$ now gives
\[
 \sum_i W_i\le\sum_i Q_i+\sum_i T_i
 \le\sum_\ell p_\ell+|V|
 \le\sum_\ell(p_\ell+d_\ell).
\]
Take the minimum over offline schedules.
\end{proof}

\begin{theorem}[Deterministic oracle bound]
\label{thm:det}
Under Definition~\ref{def:oracle}, with monotone $g$ but without requiring subadditivity of $g$, \Threshold\ is strictly $(\rho+1)$-competitive.
\end{theorem}

\begin{proof}
At a service of batch $A_i$, the trigger gives $g(A_i)=W_i$. Its cost is at most $\rho g(A_i)+W_i=(\rho+1)W_i$. 
Sum this inequality and apply Lemma~\ref{lem:detcertificate}.
\end{proof}

The proof still uses subadditivity of the true joint cost $\kappa$.
The weaker oracle requirement in Theorem~\ref{thm:det} does not remove that service-model assumption. 
The algorithm makes at most one value query per arrival epoch; a value and its returned action can be cached until the next arrival or service. 
There are at most $n$ full services.

\section{A Randomized Height Algorithm}
\label{sec:random}

We now assume that $g$ is subadditive. 
Fix an input independently of the random seed. 
The algorithm maintains a virtual process that depends only on that input and on $g$. Actual services are determined by one globally sampled threshold; they never reset or otherwise alter the
virtual process.

The analysis has three steps. 
First, construct this seed-independent trajectory. 
Second, use subadditivity to certify that its total active time $L$ is at most $\OPT$. Third, choose the global threshold distribution so that the expected cost is at most $\Rrho$ per unit of virtual active time.

\subsection{Virtual blocks and threshold crossings}

Every arrived occurrence has a virtual height in $[0,1]$, initially zero.
Unfinished occurrences are grouped into consecutive blocks of equal height, ordered by arrival. 
Older blocks have strictly larger heights than newer blocks after equal-height neighbors are merged. 
At an arrival epoch, append its entire cohort as a new height-zero block and merge any equal-height neighbors. 
Between arrivals, if an unfinished block exists, only the lowest (newest) block $B$ rises, at speed
\begin{equation}
 \frac{dh}{dt}=\frac1{g(B)}.
 \label{eq:speed}
\end{equation}
When it meets the next older block, merge the two and continue using the oracle value of the union. 
At height one a block is completed and removed from the active list. 
For analysis, its occurrences retain height one. 
The process is idle only when the list is empty. 
Heights never decrease.

Sample $\Theta_\rho\in(0,1)$ once with distribution function
\begin{equation}
 q_\rho(h)=\frac{e^{h/\rho}-1}{e^{1/\rho}-1}, \qquad 0\le h\le1.
 \label{eq:cdf}
\end{equation}
Whenever the rising block crosses $\Theta_\rho$, execute $S_g(B)$.
The same threshold is used at every crossing throughout the input.
All requests in $B$ are then marked actually served, but their heights continue in the virtual process as before.

\begin{algorithm}[htbp]
\caption{\Height$(g,S_g,\rho)$}
\label{alg:height}
\begin{algorithmic}[1]
\State Sample $\Theta_\rho$ according to~\eqref{eq:cdf}; initialize an empty ordered list of virtual blocks.
\State At each arrival epoch, append its cohort at height zero; merge equal-height neighbors.
\While{the list is nonempty and no arrival intervenes}
  \State Let $B$ denote its lowest block at height $h$.
  \State Raise $B$ at speed $1/g(B)$ until the next merge, completion, arrival, or actual threshold crossing.
  \State At a crossing of $\Theta_\rho$, execute $S_g(B)$ without changing the virtual process.
  \State At equal neighboring heights, merge; at height one, remove the completed block.
\EndWhile
\end{algorithmic}
\end{algorithm}

\begin{lemma}[Full-service invariant]
\label{lem:full}
For each fixed input, with probability one, at every regular time the actually pending occurrences are precisely those with height below $\Theta_\rho$. 
At each threshold crossing the crossing block is the entire pending queue. 
Every occurrence is served exactly once by these full services.
\end{lemma}

\begin{proof}
An arriving occurrence starts below the threshold. 
Its height subsequently increases monotonically, so it crosses at most once and is served at that
crossing. 
Before the lowest block can rise to the threshold, it merges with every older unfinished block whose height is below the threshold.
All blocks still older have already crossed and contain no pending requests. 
There are no newer blocks behind the lowest block at that instant. 
Hence its crossing clears precisely the pending queue.

The virtual evolution has finitely many structural events, as shown in Section~\ref{sec:execution}. 
For a fixed input their event heights are deterministic. 
Equality of $\Theta_\rho$ with any of these heights has probability zero. 
To specify a feasible algorithm even on that null set, at the first such coincidence clear the actual queue immediately with its oracle action (if nonempty), and thereafter run \Threshold\ on future arrivals. 
This fallback does not change the expected guarantee. 
Completion at height one eventually occurs for every occurrence, and so every threshold in $(0,1)$ is crossed before its completion.
\end{proof}

Every actual service clears the queue. Therefore the sum of its maximum-delay charges equals the total amount of time for which the actual queue is nonempty. At an active time with lowest height $h$, this queue is nonempty if and only if $\Theta_\rho>h$.

\begin{figure}[htbp]
\centering
\input{figures_tikz/virtual_heights}
\caption{\CaptionHeights}
\label{fig:heights}
\end{figure}
Figure~\ref{fig:heights} separates the virtual evolution from the actual services on a two-arrival input. 
Section~\ref{sec:twoarrival} computes the schedules for arbitrary threshold values on the same input.

\subsection{Virtual active time as a lower bound}

For $x \in [0,1]$ define the layer of arrived occurrences below height $x$ by
\[
 A_x(t)=\{i:a_i\le t,\ h_i(t)<x\},
 \qquad P(t)=\int_0^1 g(A_x(t))\,dx.
\]
The integral is just a finite sum over height intervals. 
Completed occurrences have height one and do not affect it. 
Initially $P=0$.
At an arrival epoch $b_j$, let
\begin{equation}
 y_j=P(b_j^+)-P(b_j^-).
 \label{eq:insertion}
\end{equation}
Insertion only adds occurrences to layers, so $y_j\ge0$ by monotonicity.

\begin{lemma}[Potential identity]
\label{lem:potential}
During each active time interval without a structural event, $P$ decreases at rate one. 
Merges and completions do not change $P$.
If $L$ is the total virtual active time, then
\begin{equation}
 L=\sum_{j=1}^n y_j.
 \label{eq:activeidentity}
\end{equation}
\end{lemma}

\begin{proof}
Suppose the lowest block $B$ rises from $h$ to $h+dh$ without meeting another block. 
For each layer in that height interval, its only occurrences are those of $B$; afterwards the layer is empty. 
Thus $dP=-g(B)\,dh=-dt$ by~\eqref{eq:speed}. 
Merges leave all occurrence heights unchanged, and completion removes only occurrences at height one, a
measure-zero boundary of the integral. 
After the last arrival, the remaining finite blocks complete, so $P$ finally returns to zero.
Telescoping its insertion jumps and continuous decreases proves the identity.
\end{proof}

\begin{figure}[htbp]
\centering
\input{figures_tikz/layer_potential}
\caption{\CaptionPotential}
\label{fig:potential}
\end{figure}
The area representation in Figure~\ref{fig:potential} makes the local
identity $dP=-dt$ explicit. 
The next lemma bounds the insertion jumps globally using subadditivity.

\Needspace{8\baselineskip}
\begin{lemma}[Consecutive-batch certificate]
\label{lem:batchcertificate}
For every $1\le j\le k\le n$,
\begin{equation}
 \sum_{\ell=j}^k y_\ell\le g(B_{j:k})+b_k-b_j.
 \label{eq:batchcertificate}
\end{equation}
In particular, $L\le\OPT$.
\end{lemma}

\begin{proof}
Heights of earlier occurrences do not decrease, and all occurrences added between $b_j$ and $b_k$ lie in $B_{j:k}$. For every layer $x$,
\[
 A_x(b_k^+)\subseteq A_x(b_j^-)\cup B_{j:k}.
\]
Monotonicity and subadditivity therefore give
\[
 g(A_x(b_k^+))\le g(A_x(b_j^-))+g(B_{j:k}).
\]
Integrating over $x\in[0,1]$ yields
$P(b_k^+)-P(b_j^-)\le g(B_{j:k})$. 
Between these two potential values, the potential loses the active time in $[b_j,b_k]$, which is no greater than $b_k-b_j$. 
Adding back that loss proves~\eqref{eq:batchcertificate}.

Apply this inequality to all batches of an optimal partition in Lemma~\ref{lem:normal}, use $g(B) \le \kappa(B)$, and sum. 
Every insertion
increment appears once. 
By~\eqref{eq:activeidentity},
\[
 L=\sum_j y_j\le\sum_{B\in\mathcal P^*} \bigl(g(B)+\Span(B)\bigr)\le\OPT,
\]
which proves this lemma.
\end{proof}

This is the only step in the randomized analysis that uses subadditivity of $g$. It compares entire layers by set union; no assertion about marginal costs or nested optimal offline partitions is required.

\subsection{Expected cost}

\begin{theorem}[Randomized oracle bound]
\label{thm:random}
For a normalized monotone subadditive oracle satisfying
Definition~\ref{def:oracle}, \Height\ is strictly
$\Rrho=1/(1-e^{-1/\rho})$-competitive against an oblivious adversary.
\end{theorem}

\begin{proof}
Fix the finite input. Its entire virtual trajectory is deterministic.
During a phase with rising block $B$, height $h$, and $dt=g(B)\,dh$, a threshold in $(h,h+dh)$ causes one purchase of cost at most $\rho g(B)$.
The expected purchase cost in that interval is consequently at most
\[
 \rho g(B)q_\rho'(h)\,dh =\rho q_\rho'(h)\,dt.
\]
By Lemma~\ref{lem:full}, expected actual busy time in the same interval is $(1-q_\rho(h))\,dt$. 
The sum of these two rates is constant:
\begin{equation}
 \rho q_\rho'(h)+1-q_\rho(h) = \frac{e^{h/\rho}+e^{1/\rho}-e^{h/\rho}}{e^{1/\rho}-1} =\Rrho.
 \label{eq:rate}
\end{equation}
During virtual idle time the actual queue is also empty. 
Integrating over all active phases gives $\E[\ALG]\le\Rrho L\le\Rrho\OPT$ by Lemma~\ref{lem:batchcertificate}. 
The shared threshold correlates services, but linearity of expectation suffices: the virtual trajectory being integrated is independent of that threshold.
\end{proof}

The distribution in~\eqref{eq:cdf} can also be derived by equalizing the two expected rates. Solving $\rho q'(h)+1-q(h)=R$ subject to $q(0)=0$ and $q(1)=1$ gives exactly $q_\rho$ and $R=\Rrho$.
For an exact oracle with minimum-cost realizations ($\rho=1$), the purchase-rate inequality is an equality, so the proof gives the useful identity
\begin{equation}
 \E[\ALG]=\frac{e}{e-1}L.
 \label{eq:exactidentity}
\end{equation}

\subsection{A two-arrival example}
\label{sec:twoarrival}

Suppose a single request type costs one to serve, and requests arrive at times $0$ and $1/2$. The first virtual block reaches height $1/2$ when the second arrives. The new block then rises alone until time $1$, when the blocks merge at height $1/2$. Their union reaches height one at time $3/2$. Thus $L=3/2=\OPT$.

If $0<\Theta_1<1/2$, the two actual services occur at times $\Theta_1$ and $1/2+\Theta_1$, for total cost $2+2\Theta_1$. 
If $1/2<\Theta_1<1$, there is a single service at time $1/2+\Theta_1$, for cost $3/2+\Theta_1$. 
These different schedules use the same virtual trajectory, and~\eqref{eq:exactidentity} gives their average cost as $\tfrac32 e/(e-1)$. Resetting heights at the first actual service changes the construction and invalidate the trajectory-based proof.

\subsection{Finite execution and computational model}
\label{sec:execution}

A \emph{virtual phase} is a maximal positive-duration interval on which the rising block is fixed and no arrival occurs. 
An actual threshold crossing is not a virtual phase boundary. Each of the $n$ arrival epochs
creates one block. 
A merge reduces the number of active blocks by one, and a completion does the same. 
In every finite prefix, the number of merges plus completions is at most the number of blocks created, hence at most $n$. 
After the last arrival, each nonempty active configuration reaches its next merge or completion in finite time, because its speed is positive and fixed until that event. 
Thus the process terminates; there is no accumulation of infinitely many events. 
At termination, if $M$ and $D$ count merges and completions, then $M + D = n$, $D \ge 1$, and $M \le n-1$ for a nonempty input.

Completion is possible only when one unfinished block remains: any older unfinished block would have strictly greater height, still below one. 
Therefore a completion cannot start an uncounted active phase.
Each phase starts at an arrival or a merge, and there are at most $n+M\le2n-1$ phases. 
Almost surely, at most $2n-1$ combined value/realization oracle calls suffice, caching the action associated with each phase, and there are at most $n$ actual services. 
The null-set fallback specified in Lemma~\ref{lem:full} also uses $O(n)$ calls and at most $n$ services.

With rational batch values and arrival times, the virtual event calculations use rational arithmetic and a polynomial number of arithmetic operations, in addition to the batch oracle. 
Maintaining explicit block supports adds polynomial overhead. 
For the exact randomized ratio we use the usual continuous-time, real-arithmetic randomization model: one may sample $V$ uniformly in $(0,1)$ and set
\[
        \Theta_\rho=\rho\log\bigl(1+(e^{1/\rho}-1)V\bigr).
\]
This assertion is not a finite-precision sampling theorem. 
A bit-model implementation would need a specified approximation to the distribution and a separate accounting of the resulting error. 
The deterministic algorithm and the offline dynamic program have no such sampling issue.

%% file: figures_tikz/virtual_heights.tex
\begin{tikzpicture}[vc figure, x=6cm, y=3.3cm]
  \node[vc panel] at (-.20,1.23) {(a) One fixed virtual trajectory};
  \foreach \t in {.5,1,1.5} \draw[vc guide] (\t,0)--(\t,1.05);
  \draw[vc guide] (0,.5)--(1.58,.5);
  \draw[vc guide] (0,1)--(1.58,1);
  \draw[vc axis] (0,0)--(1.64,0) node[right,vc note] {$t$};
  \draw[vc axis] (0,0)--(0,1.08) node[above,vc note] {height};
  \foreach \t/\lab in {0/0,.5/{1/2},1/1,1.5/{3/2}}
    \node[vc note,below=4pt] at (\t,0) {$\lab$};
  \foreach \h/\lab in {.5/{1/2},1/1}
    \node[vc note,left=4pt] at (0,\h) {$\lab$};

  \draw[vcmuted,densely dotted,line width=.7pt] (0,.25)--(1.58,.25);
  \draw[vcmuted,densely dotted,line width=.7pt] (0,.75)--(1.58,.75);
  \node[vc note,anchor=west] at (1.64,.25) {$\theta=1/4$};
  \node[vc note,anchor=west] at (1.64,.75) {$\theta=3/4$};
  \draw[vcblue,line width=1.5pt] (0,0)--(.5,.5)--(1,.5);
  \draw[vcorange,dashed,line width=1.5pt] (.5,0)--(1,.5);
  \draw[vcink,line width=2pt] (1,.5)--(1.5,1);
  \node[vc vertex,fill=vcink] at (1,.5) {};
  \node[vc note,anchor=west] at (1.06,.43) {merge};
  \node[vc note,vcblue] at (.69,.6) {$\{r_1\}$ waits};
  \node[vc note,vcorange,anchor=west] at (.72,.11) {$\{r_2\}$ rises};
  \node[vc note,vcink,anchor=west] at (1.07,.94) {$\{r_1,r_2\}$};
  \foreach \t/\h in {.25/.25,.75/.25,1.25/.75}
    \node[circle,draw=vcink,fill=white,minimum size=5.4pt,inner sep=0pt,line width=.8pt]
      at (\t,\h) {};

  \node[vc panel] at (-.20,-.38) {(b) Actual services selected by the threshold};
  \begin{scope}[yshift=-.32cm]
  \node[vc note,anchor=east] at (-.035,-.66) {$\theta=1/4$};
  \fill[vcblue!16] (0,-.70) rectangle (.25,-.60);
  \fill[vcorange!20] (.5,-.70) rectangle (.75,-.60);
  \draw[vc axis] (0,-.70)--(1.59,-.70);
  \node[vc arrival] at (0,-.70) {};
  \node[vc arrival] at (.5,-.70) {};
  \foreach \t in {.25,.75} {
    \draw[vcink,line width=.7pt] (\t,-.70)--(\t,-.56);
    \node[vc service] at (\t,-.56) {};
  }
  \node[vc note,above=3pt] at (.25,-.56) {$1/4$};
  \node[vc note,above=3pt] at (.75,-.56) {$3/4$};
  \node[vc note,anchor=west,align=left] at (1.64,-.65)
    {two services\\cost $2+1/2=5/2$};

  \node[vc note,anchor=east] at (-.035,-1.08) {$\theta=3/4$};
  \fill[vcblue!12] (0,-1.12) rectangle (1.25,-1.02);
  \draw[vc axis] (0,-1.12)--(1.59,-1.12);
  \node[vc arrival] at (0,-1.12) {};
  \node[vc arrival] at (.5,-1.12) {};
  \draw[vcink,line width=.7pt] (1.25,-1.12)--(1.25,-.98);
  \node[vc service] at (1.25,-.98) {};
  \node[vc note,above=3pt] at (1.25,-.98) {$5/4$};
  \node[vc note,anchor=west,align=left] at (1.64,-1.07)
    {one service\\cost $1+5/4=9/4$};
  \foreach \t/\lab in {0/0,.5/{1/2},1/1,1.5/{3/2}}
    \node[vc note,below=4pt] at (\t,-1.12) {$\lab$};
  \node[vc note,align=center] at (.8,-1.39)
    {Dots: arrivals. Squares: services. Shaded intervals: actual waiting.};
  \end{scope}
\end{tikzpicture}

%% file: figures_tikz/layer_potential.tex
\begin{tikzpicture}[vc figure, x=1cm, y=4cm]
  \node[vc panel] at (-.55,1.25) {(a) Only the lowest block rises};
  \draw[vc axis] (0,0)--(0,1.1) node[above,vc note] {height};
  \draw[vc axis] (0,0)--(5.3,0);
  \foreach \y in {.2,.6,.85,1} \draw[vc guide] (0,\y)--(5.05,\y);
  \node[vc note,left=3pt] at (0,1) {$1$};
  \node[vc note,left=3pt] at (0,.2) {$h$};
  \draw[vcblue!45,line width=1pt] (1,0)--(1,.85) (2.75,0)--(2.75,.6);
  \node[draw=vcblue,fill=vcblue!10,minimum width=24pt,minimum height=15pt]
    at (1,.85) {$B_1$};
  \node[draw=vcblue,fill=vcblue!10,minimum width=24pt,minimum height=15pt]
    at (2.75,.6) {$B_2$};
  \draw[vcblue!45,line width=1pt] (4.5,0)--(4.5,.2);
  \node[draw=vcblue,fill=vcblue!10,minimum width=24pt,minimum height=15pt]
    at (4.5,.2) {$B_3$};
  \draw[vc guide] (4.5,.34)--(5.13,.34);
  \draw[vcorange,-{Stealth[length=4pt]},line width=1pt] (5.08,.2)--(5.08,.34);
  \node[vc note,anchor=west,vcorange] at (5.17,.27) {$dh$};
  \node[vc note] at (1,-.10) {oldest};
  \node[vc note] at (4.5,-.10) {newest};
  \node[vc note] at (2.7,-.23) {arrival order};

  \begin{scope}[xshift=7.45cm]
    \node[vc panel] at (-.65,1.25) {(b) Potential is the shaded area};
    \fill[vcblue!10] (0,.2)--(2,.2)--(2,.6)--(3.3,.6)--(3.3,.85)
                         --(4.55,.85)--(4.55,1)--(0,1)--cycle;
    \fill[vcorange!23] (0,.2) rectangle (2,.34);
    \fill[pattern=north east lines,pattern color=vcorange!70]
      (0,.2) rectangle (2,.34);
    \draw[vcblue,line width=1.1pt] (0,.2)--(2,.2)--(2,.6)--(3.3,.6)
                         --(3.3,.85)--(4.55,.85)--(4.55,1)--(0,1);
    \draw[vc axis] (0,0)--(0,1.1) node[above,vc note] {$x$};
    \draw[vc axis] (0,0)--(5.02,0) node[right,vc note] {$g(A_x)$};
    \node[vc note,left=3pt] at (0,.2) {$h$};
    \node[vc note,left=3pt] at (0,.34) {$h+dh$};
    \node[vc note,left=3pt] at (0,1) {$1$};
    \draw[vc guide] (0,.34)--(2,.34);
    \node[vc note] at (1,.46) {$A_x=B_3$};
    \node[vc note] at (1.65,.72) {$A_x=B_2\cup B_3$};
    \node[vc note] at (2.27,.925) {$A_x=B_1\cup B_2\cup B_3$};
    \draw[vcorange,-{Stealth[length=4pt]},line width=.65pt]
      (3.13,.31)--(2.12,.27);
    \node[vc note,anchor=west,align=left,vcorange] at (3.15,.31)
      {removed strip\\area $g(B_3)\,dh$};
    \node[vc note] at (2.38,-.16) {$P(t)=\displaystyle\int_0^1g(A_x(t))\,dx$};
  \end{scope}

  \node[fill=vclight,rounded corners=2pt,inner xsep=12pt,inner ysep=6pt]
    at (6.5,-.47)
    {$\displaystyle \frac{dh}{dt}=\frac1{g(B_3)}
      \quad\Longrightarrow\quad dP=-g(B_3)\,dh=-dt$};
\end{tikzpicture}

%% file: sections_subadditive/lower_bounds.tex
\section{Matching Lower Bounds}
\label{sec:lower}

Fix any request type $u$ whose minimum service cost is $c>0$ and allow repeated occurrences of that type. Restrict inputs to $u$ alone. 
Every useful service clears the entire queue and costs at least $c$, and a minimum action realizes cost $c$. 
Dividing costs and times by $c$ reduces the timing problem to one with unit purchase cost. 
Accordingly, all lower bounds in this section apply to every fixed nontrivial system in our model. 
For vertex cover, the simplest such system is a single edge with unit-cost endpoints.

\subsection{Deterministic algorithms}

\begin{theorem}
\label{thm:detlower}
No deterministic algorithm has competitive ratio less than $2$, even with a fixed additive constant and only one request type.
\end{theorem}

\begin{proof}
Fix an algorithm and an integer $m\ge2$. Release the first request at zero. 
After the algorithm serves request $i$, wait $\varepsilon=1/(m-1)$ and release request $i+1$, stopping after $m$ requests. 
An algorithm that fails to serve a request in finite time has no finite guarantee. 
Otherwise let $x_i$ be the waiting time of request $i$ and $X=\sum_{i=1}^m x_i$. 
Each useful service serves exactly one request, so $\ALG\ge m+X$.

Offline may serve each request immediately, at cost $m$, or serve all requests at the last arrival, at cost
$1+\sum_{i<m}x_i+(m-1)\varepsilon\le2+X$. Hence
\[
 \frac{\ALG}{\OPT} \ge \frac{m+X}{\min\{m,2+X\}}\ge2-\frac2m.
\]
For completeness, when $X\le m-2$ the fraction with denominator $2+X$ is minimized at $X=m-2$, and when $X\ge m-2$ the fraction with denominator $m$ is at least the same value. 
For any $1\le R<2$, these two cases also give
\begin{equation}
 \ALG-R\OPT\ge(2-R)m-2.
 \label{eq:detgap}
\end{equation}
Letting $m$ grow excludes a fixed additive constant. 
For a fixed deterministic algorithm the service-dependent construction can be simulated in advance. 
It therefore yields one fixed finite input with strictly increasing arrivals, as required for a deterministic lower bound.
\end{proof}

\subsection{Randomized algorithms against an oblivious adversary}

For a one-type input with successive gaps $g_1, \ldots, g_N$, the offline normal form gives
\begin{equation}
 \OPT=1+\sum_{i=1}^N\min\{g_i,1\}.
 \label{eq:onetypeopt}
\end{equation}
Indeed, keeping two consecutive arrivals in one batch contributes their gap to the total batch spans, while cutting that gap contributes one additional purchase. 
The choices are independent across gaps.

\begin{theorem}
\label{thm:randomlower}
No randomized algorithm against an oblivious adversary has competitive ratio less than $e/(e-1)$, even with a fixed additive constant and only one request type.
\end{theorem}

\begin{proof}
Fix $N\ge1$. Release $N+1$ requests, the first at zero and with independent successive gaps distributed as $G$, where
\[
 f_G(g)=e^{-g}\quad(0<g<1),\qquad \Prob(G=2)=1/e.
\]
The density and atom have total mass one. This input distribution is chosen independently of the algorithm. Direct integration gives
\[
 \E[G]=1,\qquad \E[\min\{G,1\}]=1-1/e,
 \qquad \E[\OPT]=1+N(1-1/e).
\]

Fix the algorithm's random seed and condition on its history just before processing an arrival. 
In the hypothetical continuation with no further arrivals, let $T\in[0,\infty]$ be the time from this arrival to the first useful service. 
Before the next arrival the real execution agrees with that continuation. 
Within this gap, the queue is nonempty for at least $\min\{T,G\}$ time units and a purchase costing at least one occurs if $T<G$. 
Thus the contribution of this gap to busy time and purchases is at least
\[
 \min\{T,G\}+\ind\{T<G\}.
\]
The new gap is independent of the conditioned history and of the fixed seed. For every $T$,
\begin{equation}
 \E[\min\{T,G\}+\ind\{T<G\}]=
 \begin{cases}
  1,&0\le T\le1,\\
  1+(T-1)/e,&1<T<2,\\
  1,&T\ge2,
 \end{cases}
 \label{eq:gapcost}
\end{equation}
where the final case includes $T=\infty$. 
For $0\le T\le1$, the two terms are $\int_0^T e^{-x}\,dx$ and $e^{-T}$. 
For $1<T<2$, they are $1-1/e+(T-1)/e$ and $1/e$.
For $T\ge2$, pending time alone has expectation $\E[G]=1$ and the strict purchase indicator is zero. 
In particular the atom at $G=2$ causes no exception.

On a one-type input every useful service clears the queue, so total maximum delay equals total busy time. 
Contributions assigned to distinct gaps are disjoint. 
A purchase tied with the next arrival is deliberately excluded by $T<G$ and can instead be assigned to the next gap, or to the final interval. 
After the last arrival at least one more useful purchase is necessary, at cost at least one. It follows that
\[
 \E[\ALG]\ge N+1.
\]
This holds for every fixed seed and hence also after averaging over seeds.
If a randomized algorithm satisfied a ratio $R$ with additive constant $\beta$ on every fixed input, averaging that guarantee over the above distribution would give
\[
 N+1\le R\bigl(1+N(1-1/e)\bigr)+\beta.
\]
For $R<e/(e-1)$ the difference between the left side and the term multiplied by $R$ diverges to infinity with $N$, a contradiction.
\end{proof}

The distribution above has a continuous part and one atom; it is not a finite-support distribution. 
The proof needs only finitely many requests per input and the displayed expectations. 
No minimax theorem or discretization claim is used.

\begin{corollary}[Optimal exact-oracle ratios]
\label{cor:optimal}
On every fixed static service system satisfying the model and containing a positive-cost request type, the optimal deterministic competitive ratio is $2$, and the optimal oblivious-adversary randomized competitive ratio is $e/(e-1)$, when batch optimization is unrestricted.
\end{corollary}

\begin{proof}
Use $g=\kappa$, $\rho=1$, and minimum-cost realizations in
Theorems~\ref{thm:det} and~\ref{thm:random}. 
Apply the lower bounds above to any positive-cost type and rescale by its minimum cost.
\end{proof}

\subsection{Feedback-dependent arrivals}
\label{sec:adaptive}

Here an adversary may observe completed services before choosing later arrival times. 
We compare expected online cost to the expected optimum computed in hindsight for the realized input:
\[
 \E[\ALG(\sigma)]\le R \, \E[\OPT(\sigma)]+\beta.
\]
This explicitly differs from the fixed-input, oblivious convention.

\begin{proposition}
\label{prop:adaptive}
Under this feedback model and benchmark, no randomized algorithm has ratio less than $2$ with a fixed additive constant. 
Exact-oracle \Threshold\ attains $2$, so the bound is optimal.
\end{proposition}

\begin{proof}
Use the adversary from Theorem~\ref{thm:detlower}, releasing the next request $\varepsilon$ after the observed service of the current one.
It need not see the random seed. 
For every realization of that seed, $\ALG\ge m+X$ and $\OPT\le\min\{m,2+X\}$ hold pathwise. 
Thus \eqref{eq:detgap} holds pathwise and remains true after taking expectations; in this normalization $\E[\OPT]\le m$, even if the online expected cost is infinite. 
Letting $m$ grow excludes every $R<2$.
If the algorithm does not serve in finite time with probability one, it already fails a finite expected guarantee. 
The deterministic upper bound holds for every realized finite input and therefore supplies the matching expected upper bound as well.
\end{proof}

%% file: sections_subadditive/applications.tex
\section{Covering Applications}
\label{sec:applications}

\subsection{Weighted vertex cover}
\label{sec:vc}

Let $G=(V,E)$ denote a finite loopless undirected graph with vertex costs $c_v>0$. 
A request type is an edge. 
A service purchases $S\subseteq V$ at cost $\sum_{v\in S}c_v$ and serves every pending request with an endpoint in $S$. 
Vertices are repurchased when needed; a service incurs one maximum-delay charge for all requests it serves, not one charge per purchased vertex. Write
\[
 \tau(A)=\min\left\{\sum_{v\in S}c_v: S\cap e\ne\varnothing\text{ for every }e\in A\right\}.
\]
This is $\kappa$ for the graph service model. 
It is normalized and monotone. 
The union of a cover of $A$ and a cover of $B$ covers $A\cup B$ at cost at most their sum, so $\tau$ is subadditive.

\paragraph{Exact bipartite oracle.}
Suppose $G=(L\cup R,E)$ is bipartite. 
For a queried edge set $A$, form a network with arcs $s\to u$ of capacity $c_u$ for $u\in L$, arcs $v\to t$ of capacity $c_v$ for $v\in R$, and arcs $u\to v$ of capacity $M>\sum_{w\in V}c_w$ for $uv\in A$. 
There is a cut of capacity less than $M$. 
Hence a minimum cut, with source side $X$, cannot leave an $A$-edge from $L\cap X$ to $R\setminus X$. The set
\[
 (L\setminus X)\cup(R\cap X)
\]
is therefore a cover whose cost equals that cut's capacity. 
Conversely, every cover gives a cut of the same cost by placing its left vertices on the sink side and its right vertices on the source side. 
Thus min-cut computes $\tau(A)$ and an attaining cover~\cite{schrijver2003}.

\begin{corollary}
\label{cor:bipartite}
Weighted bipartite vertex cover with maximum delay has polynomial-time exact batch implementations with strict ratios $2$ deterministically and $e/(e-1)$ randomly against an oblivious adversary. 
These ratios are optimal even on a fixed graph containing just one edge. 
Its offline optimum is polynomial-time computable.
\end{corollary}

The unrestricted exact-oracle bounds hold on general graphs as well; bipartiteness is used for computation, not for the timing analysis.

\paragraph{Fractional oracle on general graphs.}
Define
\begin{equation}
 \tau_f(A)=\min\left\{\sum_{v\in V}c_vx_v:
       x_u+x_v\ge1\ (uv\in A),\ x\ge0\right\}.
 \label{eq:vclp}
\end{equation}
It is a normalized monotone lower bound on $\tau$. 
If $x$ and $z$ are optimal for $A$ and $B$, their coordinatewise maximum is feasible for $A\cup B$ and costs at most their summed costs. Therefore
\[
 \tau_f(A \cup B) \le \tau_f(A) + \tau_f(B).
\]
Given an optimal $x$, let $S_f(A)=\{v:x_v\ge1/2\}$. Each edge has a selected endpoint, and
\[
 c(S_f(A)) \le 2\sum_v c_vx_v=2 \, \tau_f(A).
\]
Thus $g=\tau_f$ and $S_g=S_f$ satisfy the full oracle contract with $\rho=2$~\cite{williamson2011}.

\begin{corollary}
\label{cor:general}
Weighted vertex cover with maximum delay on general graphs has polynomial-time batch implementations with strict competitive ratios
\[
 3 \quad\text{and}\quad \frac{\sqrt e}{\sqrt e-1}
\]
deterministically and against an oblivious adversary, respectively.
The same batch routine gives an offline $2$-approximation.
\end{corollary}

These statements use the execution convention of
Section~\ref{sec:execution}. Appendix~\ref{sec:tightvc} gives odd-cycle
inputs on which the specific fractional algorithms approach their
displayed competitive upper bounds.

\begin{proposition}[Static approximation inherited from online algorithms]
\label{prop:statictransfer}
A strictly $R$-competitive online algorithm for vertex cover with maximum delay yields an $R$-approximation for static weighted vertex cover by running it on an all-at-zero input and taking the union of purchased vertices. 
The same statement holds in expectation for randomized
algorithms. 
It is a polynomial-time reduction when the online execution and its service outputs can be simulated in polynomial time.
\end{proposition}

\begin{proof}
Request every edge at time zero. 
As in Proposition~\ref{prop:offlinecomplexity}, the online instance's offline optimum is the static cover optimum. 
The purchased union covers every edge and has cost no greater than total online cost. 
Apply the competitive guarantee, or its expectation. 
The additional simulation qualification is needed to infer a computational approximation theorem from an abstract continuous-time algorithm.
\end{proof}

This reduction gives a way to transfer static hardness assumptions, but does not prove that the factors in Corollary~\ref{cor:general} are best possible among polynomial-time online algorithms.

\subsection{Rank-bounded hypergraph vertex cover}

Let vertices have positive costs $c_v$, and let a request type be a nonempty hyperedge $e\subseteq V$ with
$|e|\le r$, where $r\ge1$ is fixed for the instance. 
A service purchases vertices and serves all pending hyperedges it intersects, paying one joint maximum delay. 
The integral joint cost $\tau_r$ is again normalized, monotone, and subadditive. Consider
\begin{equation}
 g_r(A)=\min\left\{\sum_v c_vx_v: \sum_{v\in e}x_v\ge1\ (e\in A),\ x\ge0\right\}.
 \label{eq:hyperlp}
\end{equation}
The coordinatewise-maximum argument proves subadditivity, and monotonicity is immediate from the constraints. 
Select all vertices with $x_v\ge1/r$ in an optimum solution. 
Every hyperedge is hit: otherwise its at most $r$ coordinates would sum to less than one. Its cost is at most $r g_r(A)$.

\begin{corollary}
\label{cor:hyper}
Rank-$r$ hypergraph vertex cover with maximum delay has strict ratios $r+1$ deterministically and $1/(1-e^{-1/r})$ randomly, using polynomial-time LP-based batch routines. 
Its offline problem has an $r$-approximation using the same routine. With exact integral batch optimization, the ratios improve to $2$ and $e/(e-1)$, respectively.
\end{corollary}

For each fixed $r$, cyclic hypergraphs make the fractional implementations asymptotically attain the first two bounds (Appendix~\ref{sec:tighthyper}). 
These examples do not assert computational optimality for the hypergraph problem.

\subsection{Family-service weighted Set Cover}
\label{sec:family}

Let $U$ denote a universe of $m \ge 1$ types and $\mathcal F$ a finite explicit family of sets with positive costs $c_S$. 
Every element is coverable.
A service may purchase a subfamily $\mathcal H\subseteq\mathcal F$ at cost $\sum_{S\in\mathcal H}c_S$. 
It serves all pending elements in $\bigcup_{S\in\mathcal H}S$ and incurs \emph{one} maximum waiting time over those occurrences. 
The minimum joint cost is the ordinary weighted set-cover value for the requested subset. 
This is subadditive because two purchased families can be combined into one service.

Use the fractional value
\begin{equation}
 F(A)=\min\left\{\sum_{S\in\mathcal F}c_Sx_S: \sum_{S\ni u}x_S\ge1\ (u\in A),\ x\ge0\right\}.
 \label{eq:sclp}
\end{equation}
It is normalized, positive on nonempty batches, monotone, subadditive by coordinatewise maximum, and no greater than the integral joint cost.
For realization, run weighted greedy set cover on the distinct types of $A$: repeatedly choose a set of minimum cost per newly covered type, with fixed tie-breaking. 
The following standard analysis makes explicit that the guarantee is relative to $F(A)$, not merely to the integral optimum~\cite{chvatal1979,williamson2011}.

\Needspace{6\baselineskip}
\begin{lemma}[Greedy realization relative to the LP]
\label{lem:greedy}
The greedy family has cost at most $H_m F(A)$.
\end{lemma}

\begin{proof}
When greedy selects a set at cost per newly covered element $p$, assign price $p_u = p$ to each type it newly covers. 
The sum of these prices is the greedy purchase cost. Fix any set $S\in\mathcal F$. 
When $k$ of its requested types remain uncovered, greedy's chosen price is at most $c_S/k$, since $S$ itself is available. 
Order the types of $A\cap S$ by when they are covered, breaking simultaneous ties arbitrarily. 
Summing the resulting bounds gives
\[
 \sum_{u\in A\cap S}p_u\le H_{|A\cap S|}c_S\le H_m c_S.
\]
Within a tied group, the bound with the larger pre-selection value of $k$ is only stronger than the individual harmonic bounds. 
Thus $(p_u/H_m)_{u\in A}$ is a feasible dual solution to~\eqref{eq:sclp}.
Weak duality implies $\sum_u p_u\le H_m F(A)$.
\end{proof}

\begin{corollary}
\label{cor:family}
Family-service weighted Set Cover with maximum delay has polynomial-time batch implementations with strict competitive ratios
\[
 H_m+1 \quad\text{and}\quad \frac1{1-e^{-1/H_m}},
\]
and an offline $H_m$-approximation. 
With exact batch optimization its optimal ratios are $2$ and $e/(e-1)$.
\end{corollary}

Apply the oracle theorems with $g=F$ and $\rho=H_m$. Importantly, the height speed uses the fractional optimum, not the value returned by greedy. 
The greedy value alone is not being assumed monotone or subadditive. 
Appendix~\ref{sec:tightsc} gives a single-batch example approaching both implementation bounds for each fixed $m \ge 2$.

\section{Relation to Multi-Level Aggregation}
\label{sec:mla}

The companion manuscript~\cite{lu2026mla}, Theorem~6.1, gives the exact competitive constants for realizable submodular service systems. 
We relate its tree model and proof structure to the present framework.

\subsection{Rooted-tree joint costs fit the framework}

For a rooted tree, let $w_v\ge0$ denote the cost of vertex $v$ and let $T_v$ denote the descendant subtree rooted at $v$. 
A rooted connected service reaching a set $A$ of request locations must contain the union of their root paths. Its minimum cost is
\begin{equation}
 \kappa_T(A)=\sum_{v}w_v\ind\{A\cap T_v\ne\varnothing\}.
 \label{eq:treecost}
\end{equation}
The same formula applies to edge costs by indexing descendant subtrees by their incoming edges. Each summand is a nonnegative weighted coverage function, hence is normalized, monotone, and submodular.

Every normalized nonnegative submodular function $f$ is subadditive, since
\[
 f(A\cup B)\le f(A)+f(B)-f(A\cap B)\le f(A)+f(B).
\]
Root-path unions are executable whenever the corresponding requests are pending. 
Under the same instantaneous, nonpersistent, joint-maximum service semantics, MLA therefore satisfies our service contract. 
Free root or zero-path types are removed as in Section~\ref{sec:model}.
The exact-oracle algorithms give the same ratios $2$ and $e/(e-1)$, independently of tree depth. 
This recovers the competitive guarantees, not the particular DP-Envelope schedule.

\subsection{What the subadditive extension changes}

The companion analysis uses diminishing returns and an interval-cost quadrangle inequality to obtain causal DP-Envelope deadlines and nested threshold-dependent partitions. 
Our proof instead compares layer sets by inclusion and union: Lemma~\ref{lem:batchcertificate} uses
\[
 g(X\cup B)\le g(X)+g(B)
\]
to certify virtual active time as a lower bound on $\OPT$. 
The deterministic algorithm uses direct interval charging. 
Thus the competitive guarantees carry over without requiring the envelope or partition-nesting properties of the companion algorithm.

\begin{proposition}[Strict separation of cost assumptions]
\label{prop:nonsubmod}
The minimum batch cost of bipartite vertex cover need not be submodular.
Consequently, within the common finite-type static-service model, monotone subadditivity is a strictly weaker cost assumption than submodularity.
\end{proposition}

\begin{proof}
The three-edge path in Section~\ref{sec:intro} and
Figure~\ref{fig:separation} supplies the violation of submodularity.
Subadditivity and realizability follow from the vertex-cover service model in Section~\ref{sec:vc}.
\end{proof}

On this bipartite path, fractional and integral cover values agree, so the separation also holds for the fractional cost function. 
The approximation-oracle bounds provide a further distinction: they allow the subadditive lower-bound value and the cost of its integral realization to differ.

\section{Why Joint Service Semantics Matter}
\label{sec:boundary}

The cost in~\eqref{eq:servicecost} attaches one maximum delay to a whole action. 
This permits a union of vertex purchases, hypergraph hitting vertices, or set-cover sets to be treated as one service. 
The offline normal form needs such a batch to be realizable for its scalar cost; the online algorithms need the same property when they clear a queue.

An atomic set-purchase model has a different accounting rule. 
If a subfamily $\mathcal H$ is purchased, each $S\in\mathcal H$ has its own served batch $B_S$ and its own delay charge. 
Its cost has the form
\[
 \sum_{S\in\mathcal H} \left(c_S+\max_{i\in B_S}(t-a_i)\right),
\]
with empty batches omitted. This need not equal or be bounded by the family-service cost with just one maximum.

For example, take two request types $u,v$ and only the two singleton sets, each costing one. 
Suppose both requests arrive at zero and both are served at time $t>0$. 
A family service costs $2+t$, but two atomic purchases cost $2+2t$. 
If a single atomic purchase is the allowed action, there is not even an action covering $\{u,v\}$, although each singleton is coverable. 
Treating the two purchases as an action of spatial cost two would not repair the mismatch in the delay term. 
This is a direct failure of the service contract, not evidence that the static set-cover optimum ceases to be subadditive.

Accordingly, Corollary~\ref{cor:family} concerns the family-service problem only. 
It does not solve atomic rebuy Set Cover, show that its competitive ratios are the same, or establish an impossibility for algorithms tailored to that different objective. 
Time-dependent feasibility, persistent resources, action durations, and non-unit or nonlinear waiting penalties likewise require separate arguments.

%% file: sections_subadditive/tightness.tex
\section{Tight Examples for the Approximate Implementations}
\label{sec:tight}

The optimality lower bounds in Section~\ref{sec:lower} apply to all algorithms. In contrast, this appendix concerns the stated realizable oracles and their fixed rounding or greedy actions. 
It demonstrates that their upper-bound analyses cannot in general be improved while leaving these implementations unchanged. 
These examples are not unconditional or complexity-theoretic lower bounds for arbitrary efficient algorithms.

\Needspace{10\baselineskip}
\subsection{A single-batch identity}

\begin{lemma}
\label{lem:onebatch}
Suppose all requests of a nonempty batch $A$ arrive at zero and no later request arrives. 
Set $C=g(A)$, $K=\kappa(A)$, and $p=c(S_g(A))$. 
The deterministic and randomized oracle algorithms have ratios, respectively,
\begin{equation}
 \frac{p+C}{K} \qquad\text{and}\qquad \frac{p+(\Rrho-\rho)C}{K}.
 \label{eq:onebatch}
\end{equation}
\end{lemma}

\begin{proof}
Offline serves the entire batch at zero for cost $K$, and subadditivity precludes a cheaper sum of purchases. 
The deterministic algorithm waits until $C$ and pays $p+C$. 
In the virtual process there is one block with height $h=t/C$ until its completion. 
The actual service is at
$C\Theta_\rho$, of cost $p+C\Theta_\rho$. Finally,
\[
 \E[\Theta_\rho]=\int_0^1(1-q_\rho(h))\,dh
     =\frac{e^{1/\rho}}{e^{1/\rho}-1}-\rho=\Rrho-\rho.
\]
Substitution proves both identities.
\end{proof}

In particular, when $p=\rho C$, the ratios are
$(\rho+1)C/K$ and $\Rrho C/K$. 
Tightness requires both expensive realization relative to $g$ and a small gap between $g$ and the true batch optimum. 
A large integrality gap alone does not imply the desired online lower example.

\subsection{Odd cycles for the general-graph LP implementation}
\label{sec:tightvc}

Take the unit-cost odd cycle on $q=2k+1$ vertices, $k\ge1$, and request all its edges at time zero. 
The fractional optimum is $C=q/2$: summing all edge constraints proves the lower bound, and $x_v=1/2$ attains it.
Moreover, this is the unique optimal solution. 
At an optimum the summed constraints are tight, so every edge satisfies $x_u+x_v=1$. 
Alternating these equalities around an odd cycle forces every coordinate to be $1/2$.

The specified rounding rule selects all vertices, at cost $p=q$. The integral cover optimum is $K=k+1=(q+1)/2$. 
Thus Lemma~\ref{lem:onebatch} with $\rho=2$ gives exactly
\[
 \frac{3q}{q+1} \qquad\text{and}\qquad \frac{\sqrt e}{\sqrt e-1}\frac{q}{q+1}.
\]
As $q$ tends to infinity through odd integers, these approach the two bounds in Corollary~\ref{cor:general}. 
Uniqueness of the fractional optimum removes dependence on LP tie-breaking, but changing the rounding procedure can change the example's behavior.

\subsection{Cyclic hypergraphs for fixed rank}
\label{sec:tighthyper}

Fix $r\ge2$ and choose $n=kr+1$ with $k\ge1$. 
Use $n$ unit-cost vertices indexed modulo $n$, and let the hyperedges be the $n$ cyclic windows
\[
 E_i=\{i,i+1,\ldots,i+r-1\}.
\]
Request all hyperedges at zero. 
Summing the LP constraints shows $C=n/r$, attained by $x_i=1/r$. 
At any optimum all window constraints are tight. 
Subtracting consecutive window equalities yields $x_i=x_{i+r}$. Since $\gcd(n,r)=1$, all coordinates are equal, so the
fractional optimum is uniquely $1/r$ at each vertex. 
Threshold rounding therefore purchases every vertex and has $p=n=rC$.

An integral cover must meet all $n$ windows. 
Each selected vertex belongs to $r$ windows, so incidence counting gives $r|S|\ge n$. 
Conversely, select vertices $0,r,2r,\ldots,kr$. 
The cyclic gaps between selected
vertices are at most $r$, so every $r$-vertex window contains a selected vertex. Therefore
\[
 K=\lceil n/r\rceil=k+1.
\]
Lemma~\ref{lem:onebatch} gives ratios
\[
 (r+1)\frac{n}{r\lceil n/r\rceil} \qquad\text{and}\qquad \frac1{1-e^{-1/r}}\frac{n}{r\lceil n/r\rceil}.
\]
For each \emph{fixed} $r$, letting $k\to\infty$ makes these approach $r+1$ and $1/(1-e^{-1/r})$. 
Fixing $r$ is part of this asymptotic claim; we do not assert absolute convergence for arbitrary sequences in which $r$ itself grows. 
For $r=1$, a one-vertex batch already attains the exact-oracle ratios by~\eqref{eq:onebatch} with $p=C=K$.

\subsection{A greedy hard batch for family-service Set Cover}
\label{sec:tightsc}

Fix $m\ge2$ and $0<\varepsilon<1-1/H_m$. 
The universe is $\{u_1,\ldots,u_m\}$. 
Provide one universal set of cost one and the singleton $\{u_j\}$ of cost $(1-\varepsilon)/j$ for each $j$. Request all elements at zero.

The integral optimum is $K=1$. The fractional optimum is also $C=1$: if $t\in[0,1]$ is the weight on the universal set, the least feasible weight on each singleton is $1-t$, yielding objective
\[
 t+(1-t)(1-\varepsilon)H_m\ge1.
\]
Values $t>1$ cannot improve the optimum. Equality is attained at $t=1$, and the strict choice of $\varepsilon$ makes this the unique optimal universal weight.

Greedy first takes $\{u_m\}$, then $\{u_{m-1}\}$, and so on. 
Indeed, when $u_1,\ldots,u_j$ remain uncovered, the universal set has cost per new element $1/j$, while $\{u_j\}$ has strictly smaller ratio $(1-\varepsilon)/j$ and is the cheapest remaining singleton. 
Thus the greedy purchase cost is $p=(1-\varepsilon)H_m$. For the algorithm parameter $\rho=H_m$,~\eqref{eq:onebatch} becomes
\[
 1+(1-\varepsilon)H_m \qquad\text{and}\qquad R_{H_m}-\varepsilon H_m.
\]
For fixed $m$, these approach $H_m+1$ and $R_{H_m}$ as
$\varepsilon\downarrow0$. Every greedy comparison is strict, so tie breaking does not affect the construction. 
For $m=1$ greedy is exact and the single-batch identity gives the corresponding exact bounds.

\paragraph{Strict ratios versus additive constants.}
The examples above establish tight suprema for the strict competitive ratios of the specified algorithms. 
If a lower example against a fixed additive constant is desired, repeat the same batch in phases separated
by more than $g(A)$. 
Both algorithms then complete their actual and virtual service before the next phase; the shared random threshold does not change the expected cost per phase. 
Offline can serve every phase at its arrival for cost $K$ per phase. 
The online cost minus $R$ times this feasible offline cost grows linearly whenever the displayed single-phase ratio exceeds $R$. 
Hence these implementation-specific lower examples also exclude smaller ratios with a fixed additive constant.